\documentclass[a4paper,UKenglish]{lipics-v2016-tr}
 
\usepackage{microtype}
\usepackage{MnSymbol}
\usepackage{stmaryrd}
\usepackage{mathtools}

\usepackage{cleveref}
\usepackage{xspace}

\usepackage{tabularx}

\usepackage{tikz}
\usetikzlibrary{arrows,decorations.pathmorphing,backgrounds,positioning,fit,
petri}

\newcommand{\fixqed}{\makeatletter\let\qed\relax\makeatother}

\title{Homogeneous Equations of Algebraic Petri Nets\footnote{This work was partially funded by the DFG Graduiertenkolleg 1651 (SOAMED).}}
\titlerunning{Homogeneous Equations of Algebraic Petri Nets}

\author[1]{Marvin Triebel}
\author[2]{Jan Sürmeli}
\affil[1]{Humboldt-Universität zu Berlin\\
  \texttt{triebel@hu-berlin.de}}
\affil[2]{Humboldt-Universität zu Berlin\\
  \texttt{suermeli@hu-berlin.de}}

  \authorrunning{M. Triebel and J. Sürmeli} 

\Copyright{Marvin Triebel and Jan Sürmeli}


\EventLogo{}
\ArticleNo{\ } 

\newcommand{\ffrom}{:}
\newcommand{\fto}{\to}
\newcommand{\nats}{\mathbb{N}}

\newcommand{\ints}{\mathbb{Z}}

\newcommand{\absval}[1]{|#1|}

\newcommand{\group}{\mathbb{G}}
\newcommand{\gplus}{\oplus}
\newcommand{\gminus}{\ominus}
\newcommand{\gzero}{0_{\group}}
\newcommand{\groupstruc}{(\group,\gplus)}
\newcommand{\termMult}{\odot}
\newcommand{\polys}[2]{#1\langle#2\rangle}
\newcommand{\ml}[2]{(#1,#2)}

\newcommand{\kvector}{\vec{k}}

\newcommand{\terms}{\Theta}
\newcommand{\gterms}{\terms_\emptyset}

\newcommand{\ar}[1]{{\mathrm{a}_{#1}}}
\newcommand{\indices}{\mathbb{I}}
\newcommand{\support}[1]{{\mathrm{supp}({#1})}}
\newcommand{\varset}{{\mathrm{VAR}}}
\newcommand{\avarset}{V}

\newcommand{\subst}[2][\sigma]{{#2}{#1}}
\newcommand{\val}[2][\sigma]{\llbracket{#2}\rrbracket_{#1}}

\newcommand{\vectors}[2][P]{\polys{#2}{\terms}^{#1}}
\newcommand{\gvectors}[2][P]{\polys{#2}{\gterms}^{#1}}

\newcommand{\apn}{APN\xspace}
\newcommand{\apns}{APNs\xspace}
\newcommand{\apnstruc}{APNS\xspace}

\newcommand{\signature}{\Sigma}

\newcommand{\psym}{P}
\newcommand{\marking}{\vec{m}}
\newcommand{\markings}{\gvectors{\ints}_{\geq 0}}

\newcommand{\tdelta}{\vec{t}^\Delta}
\newcommand{\tdeltap}{\vec{t^\prime}^\Delta}
\newcommand{\tplus}{\vec{t}^+}

\newcommand{\tminus}{\vec{t}^-}
\newcommand{\tminusp}{\vec{t^\prime}^-}
\newcommand{\step}[4]{{#1}\mathbin{\left[#2#3\right\rangle}{#4}}
\newcommand{\runldots}[8]{\step{#1}{#2}{#3}{#4}\ldots\step{#5}{#6}{#7}{#8}}

\newcommand{\bsptext}[1]{\textsf{#1}} 

\newcommand{\kterm}{\kappa}
\newcommand{\kcoeff}{\gamma}
\newcommand{\spset}{\mathcal{S}}
\newcommand{\zero}{zero\xspace}
\newcommand{\zeros}{zeros\xspace}
\newcommand{\Zero}{Zero\xspace}

\newcommand{\pre}[1]{\operatorname{pre}(#1)}

\newcommand{\lref}[1]{Lemma~\ref{#1}\xspace}

\newcommand{\dref}[1]{Definition~\ref{#1}\xspace}

\newcommand{\tref}[1]{Theorem~\ref{#1}\xspace}

\newcommand{\fref}[1]{Figure~\ref{#1}\xspace}
\newcommand{\Fref}[1]{Figure~\ref{#1}\xspace}
\newcommand{\sref}[1]{Section~\ref{#1}\xspace}

\begin{document}


\renewcommand{\footnote}[1]{\textcolor{red}{\oldfn{\textcolor{red}{#1}}}}
\maketitle

\begin{abstract}
\emph{Algebraic Petri nets} are a formalism for modeling distributed systems and algorithms, describing control and data flow by combining Petri nets and algebraic specification. 
One way to specify correctness of an algebraic Petri net model $N$ is to specify a \emph{linear equation} $E$ over the places of $N$ based on term substitution, and coefficients from an abelian group $\group$. 
Then, $E$ is \emph{valid} in $N$ iff $E$ is valid in each reachable marking of $N$. 
Due to the expressive power of Algebraic Petri nets, validity is generally undecidable. 
\emph{Stable} linear equations form a class of linear equations for which validity is decidable. \emph{Place invariants} yield a well-understood but incomplete characterization of all stable linear equations. In this paper, we provide a complete characterization of stability for the subclass of \emph{homogeneous} linear equations, by restricting ourselves to the interpretation of terms over the Herbrand structure without considering further equality axioms. 
Based thereon, we show that stability is decidable for homogeneous linear equations if $\group$ is a cyclic group.
\end{abstract}

\section{Introduction}

The formalism of \emph{algebraic Petri nets} (\apns) permits to formally model both control flow and data flow of distributed systems and algorithms, 
extending Petri nets with concepts from algebraic specification, namely a signature together with equality axioms. Thus, \apns combine the benefits of Petri nets, such as explicit modeling of concurrency and options for structural analysis, with the ability to describe data objects on a freely chosen level of abstraction. The price to pay for this expressive power is that many important behavioral properties, such as reachability of a certain marking, are undecidable. However, there are behavioral properties that can be proven based on \emph{structural} properties, such as invariants. 

In this paper, we study a particular class of behavioral properties, namely \emph{linear equations}. Intuitively, a linear equation $E$ formalizes a linear correlation between the tokens on different places, requiring that each reachable marking satisfies $E$. More formally, an \apn $N$ is defined over a signature $\Sigma$, and the tokens are ground terms over $\Sigma$. A linear equation $E$ has the form $\sum_{p\in P} \kcoeff_p \kterm_p = b_1\mu_1 + \ldots + b_n\mu_n$, where $P$ is the set of places, each $\kcoeff_p$ and $b_i$ are coefficients stemming from an abelian group, each $\kterm_p$ is a term over $\Sigma$, and each $\mu_i$ is a ground term over $\Sigma$. A marking satisfies $E$ if substituting each variable in each $\kterm_p$ with the tokens on $p$ yields an equality. \emph{Validity} of $E$ in $N$ requires that each reachable marking of $N$ satisfies $E$. Case studies have shown that this class of properties permits to formalize important behavioral properties of distributed systems and algorithms. Unfortunately, verifying the validity of $E$ in $N$ is generally infeasible. However, if $E$ is \emph{stable} then validity of $E$ becomes decidable. Stability requires the preservation of $E$ along all---not necessarily reachable---steps, that is, if a marking satisfies $E$, then firing a transition yields a marking satisfying $E$. Now, if $E$ is stable, validity of $E$ coincides with the initial marking satisfying $E$. 

\emph{Place invariants} yield a subclass of stable linear equations. Intuitively, a place invariant is a solution of a homogeneous system of linear equations given by the structure of $N$, providing the coefficients $\kcoeff_p$ and terms $\kterm_p$---the right hand side can be chosen arbitrarily. This characterization is known to be decidable but incomplete, that is, there are stable linear equations, such that the left hand side is not given by a place invariant. A decidable, complete characterization of stability---or an undecidability proof---is still an open problem. 

In this paper, we contribute to this field of study as follows: 
\begin{enumerate}
 \item We show the undecidability of validity of homogeneous equations.
 \item We provide a complete characterization of stability, restricting ourselves to
 \begin{itemize}
  \item \emph{homogeneous} linear equations, that is, $n = 1$ and $b_1 = 0$, and \item the interpretation of terms in the Herbrand structure, that is, assuming coincidence of syntax and semantics of a term, without considering further equality axioms for terms. 
 \end{itemize}
 \item We show that our characterization is decidable if the coefficients stem from a cyclic group. 
\end{enumerate}

\noindent\sref{sec:formalization} recalls required notions for equations of algebraic Petri nets. We summarize our main theorems in \sref{sec:results}, and prove these theorems in \sref{sec:undecidability} and \sref{sec:decide-stable}. We discuss related work in \sref{sec:related}, and conclude in \sref{sec:conclusion}.
Missing proofs can be found in \sref{sec:proofs} of the appendix.

\section{Formalization}
\label{sec:formalization}

\subsection{Preliminaries}

We write $\ints$ for the set of all integers, and $\nats$ denotes the set $\{0,1,2,\ldots\}$ of natural numbers including 0. Let $z\in\ints$. Then, $\absval{z}$ denote the absolute value. 

\subsubsection{Polynomials over Abelian Groups}

\emph{Polynomials over abelian groups} serve as a common algebraic base to formalize \apns and linear equations of \apns.

\begin{definition}[Abelian Group, Scalar Product]
An \emph{abelian group} $\groupstruc$ consists of a set $\group$, and an associative, commutative, binary operation $\gplus$ on $\group$ with an identity $\gzero$, and inverses $\gminus g$ for each $g\in \group$. Let $z\in\ints$ and $a\in\group$. We define the \emph{scalar product} $za\in\group$ by
\begin{equation*}
 za := \begin{cases}
       \displaystyle\bigoplus_{i=0}^z a & \mbox{if } z \geq 0\\
       \gminus (-za) & \mbox{otherwise}.
      \end{cases}
\end{equation*}
$(\group,\oplus)$ is \emph{cyclic} iff there exists $a\in\group$, such that $\group = \{za\mid z\in\ints\}$.
\end{definition}

\noindent Whenever clear from context, we simply write $\group$ for $\groupstruc$. Examples for abelian groups are the real numbers, rational numbers, integers, and the additive group $\ints/n\ints$ of integers modulo some $n\in\nats$. The group $\ints$ is infinite and cyclic, the group $\ints/n\ints$ is finite and cyclic. In contrast to that, the group of rational numbers is not cyclic.


\begin{definition}[Series, Polynomial, Monomial, Empty Polynomial]
 Let $M$ be a set, $\group$ be an abelian group, and $f \ffrom M \fto \group$ be a function. Then, $f$ is a (linear) \emph{series} over $M$ and $\group$ with \emph{support} $\support{f} := \{ m \in M \,\mid\, f(m) \not = \gzero \}$. If $\support{f}$ is finite, then $f$ is a \emph{polynomial}. We write $\polys{\group}{M}$ for the set of all polynomials over $M$ and $\group$. If $\support{f}$ is singleton, $f$ is a \emph{monomial}, and we denote $f$ by $\ml{m}{a}$ where $\support{f} = \{m\}$ and $f(m) = a$. If $\support{f} = \emptyset$, then $f$ is \emph{empty}, and we denote $f$ by $\gzero$.

\end{definition}
We lift $\gplus$ and the scalar product to $\polys{\group}{M}$ by pointwise application: 
\begin{definition}[Addition of Polynomials]
Let $M$ be a set and $\group$ be an abelian group. For $p_1,p_2\in\polys{\group}{M}$, $m\in M$, and $z\in\ints$, we define the polynomials $p_1\gplus p_2$ and $z p_1$ over $M$ and $\group$ by
\begin{align*}
 (p_1\gplus p_2)(m) &:= p_1(m)\gplus p_2(m)\enspace,\\
 (z p_1)(m) &:= zp_1(m)\enspace.
\end{align*}
\end{definition}
We lift associative binary operations from $M$ to $\polys{\group}{M}\times\polys{\ints}{M}$ by applying the Cauchy product:
\begin{definition}[Cauchy Product]
 Let $\termMult$ be an associative binary operation on a set $M$, $\group$ be an abelian group, $p_1\in\polys{\group}{M}$, and $p_2\in \polys{\ints}{M}$. We define the series $p_1\termMult p_2$ over $M$ and $G$ by 
\begin{equation*}
(p_1\termMult p_2)(m):=\bigoplus_{m=m_1\termMult m_2}\underbrace{p_2(m_2)}_{\in\ints}\underbrace{p_1(m_1)}_{\in\group}.  
\end{equation*}
\end{definition}
Because $p_1$ and $p_2$ are polynomials, the set $\support{p_1\termMult p_2} = \{m_1\termMult m_2\mid m_1,m_2\in\group,p_1(m_1)\not=\gzero,p_2(m_2)\not=0\}$ is finite, and thus $p_1\termMult p_2$ is again a polynomial over $M$ and $\group$.

\subsubsection{Terms}

For this paper, we fix a set of variables $\varset$, a non-empty, finite index set $\indices$, and a signature $\signature = (\dot{f}_i/\ar{i})_{i\in\indices}$ consisting of $|\indices|$ distinct function symbols $\dot{f}_i$ with respective arity $\ar{i}$. 

\begin{definition}[Term]
For a set $\avarset\subseteq\varset$, the set $\terms_{\avarset}$ of \emph{terms} over variables $\avarset$ is the smallest set satisfying the following conditions: 
\begin{enumerate}
 \item $\avarset \subset \terms_{\avarset}$. 
 \item Let $i\in\indices$, and $\theta_1,\ldots,\theta_{\ar{i}}\in\terms_{\avarset}$. Then, $\dot{f}_i(\theta_1,\ldots,\theta_{\ar{i}})\in\terms_{\avarset}$. 
\end{enumerate}
The elements of $\gterms$ are called \emph{ground terms}. 
\end{definition}
As usual, if $\ar{i} = 0$, we abbreviate $\dot{f}_i()$ as $\dot{f}_i$. We abbreviate the set $\terms_\varset$ of all terms as $\terms$. 

A \emph{substitution} maps each variable to a term. A substitution is an \emph{assignment} if it maps each variable to a ground term. 

\begin{definition}[Substitution, Assignment]
Every function $\sigma \ffrom \varset \to \terms$ is a substitution. Let $\theta\in\terms$. The term $\subst[\sigma]{\theta}$ is defined by: 
\begin{equation*}
 \subst[\sigma]{\theta} := \begin{cases}
                   \sigma(\theta) & \mbox{if } \theta\in\varset\\
                   \dot{f}_i(\subst[\sigma]{\theta_1},\ldots,\subst[\sigma]{\theta_{\ar{i}}}) & \mbox{if } \theta = \dot{f}_i(\theta_1,\ldots,\theta_{\ar{i}}), i\in\indices.
                  \end{cases}
\end{equation*}
If $\sigma(x)\in\gterms$ for each $x\in\varset$, then $\sigma$ is an \emph{assignment}, and we also write $\val[\sigma]{\theta}$ instead of $\subst[\sigma]{\theta}$.
\end{definition}
Obviously, if $\sigma$ is an assignment, then $\val{\theta}\in\gterms$ for all $\theta\in\terms$. 

\emph{Unification} is the problem of applying a substitution to terms, such that the resulting terms become identical. 
\begin{definition}[Unification problem, unifier, solvable]
A \emph{unification problem} $U$ is a finite subset $\{(\theta_1,\theta_1^\prime),\ldots,(\theta_n,\theta_n^\prime)\}$ of $\terms\times\terms$, also denoted by $\{\theta_1\doteq\theta_1^\prime,\ldots,\theta_n\doteq\theta_n^\prime\}$. A substitution $\sigma$ is a \emph{unifier} for $U$ iff for all $1\leq i\leq n$: $\subst{\theta_i} = \subst{\theta_i^\prime}$. If there exists a unifier for $U$, then $U$ is \emph{solvable}.
\end{definition}
It is known that every solvable unification problem has a \emph{most general unifier} (up to variants) that subsumes all other unifiers: 
\begin{lemma}
 Let $U$ be a solvable unification problem. Then, there exists a unifier $\hat{\sigma}$ for $U$, such that: For each unifier $\sigma$ for $U$, there exists a substitution $\sigma^\prime$ with $\sigma(x) = \subst[\sigma^\prime]{\hat{\sigma}(x)}$ for all $x\in\varset$.
\end{lemma}

We define a \emph{product} on terms by means of term substitution: The product of $\varrho$ and $\theta$ is defined by substituting every occurrence of any variable in $\varrho$ by $\theta$.  

\begin{definition}[Term Product]
Let $\varrho ,\theta\in \terms$ be terms, and $\sigma$ be the substitution with $\sigma(x) = \theta$ for all $x\in\varset$. Then, $\varrho\termMult\theta := \subst{\varrho}$ is the \emph{product} of $\varrho$ and $\theta$. 
\end{definition}
We observe that $\termMult$ is associative. If $\varrho\in\gterms$, then $\varrho\termMult\theta = \varrho$.

We lift substitutions from terms to polynomials over terms and abelian groups by pointwise substitution and subsequent ``simplification'' of the polynomial: 
\begin{definition}[Substitutions in Polynomials over Terms]
Let $\group$ be an abelian group, and $p\in\polys{\group}{\terms}$. Let $\sigma$ be a substitution. We define $\subst{p}\in\polys{\group}{\terms}$ by 
\begin{equation*}
\subst{p}(\theta) := \bigoplus_{\theta=\val{\theta^\prime}}p(\theta'). 
\end{equation*}
If $\sigma$ is an assignment, we also write $\val{p}$ instead of $\subst{p}$.
\end{definition}
We observe $\subst{(\varrho\termMult\theta)} = \varrho\termMult\subst{\theta}$ for all $\varrho,\theta\in\terms$, and $\subst{(p_1\termMult p_2)} = p_1\termMult\subst{p_2}$ for all $p_1,p_2\in\polys{\group}{\terms}$. Moreover, if $\sigma$ is an assignment then $\support{\val{p}}\subseteq\gterms$. 

\subsubsection{Vectors}

In this paper, a \emph{$P$-vector} is a mapping from a set $P$ into polynomials over terms and an abelian group. 
\begin{definition}[$P$-vector]
Let $P$ be a set, $\groupstruc$ be an abelian group, and $\kvector\ffrom P\fto\polys{\group}{\terms}$. Then, $\kvector$ is a \emph{$P$-vector} over $\group$. We write $\vectors{\group}$ for the set of all $P$-vectors over $\group$. If $\kvector(p)$ is a monomial for each $p\in P$, then $\kvector$ is \emph{simple}. If $\group = \ints$, and $\kvector\geq 0$ ($\kvector\leq 0$), then $\kvector$ is \emph{semi-positive} (\emph{semi-negative}). 
\end{definition}

In order to simplify notation, we lift the basis notions from polynomials to $P$-vectors: 
\begin{definition}[$P$-vectors: Support, emptiness, addition, Cauchy product, and assignments]
Let $P$ be a set, $\groupstruc$ be an abelian group, $\kvector,\kvector_1,\kvector_2\in\vectors{\group}$, and $\kvector^\prime\in\vectors{\ints}$. 
\begin{itemize}
 \item $\support{\kvector} := \bigcup_{p\in P} \support{\kvector(p)}$ is the \emph{support} of $\kvector$. 
 \item If $\kvector(p) = \gzero$ for all $p\in P$, then $\kvector$ is the empty $P$-vector, also denoted by $\gzero$.  
 \item We define $(\kvector_1\gplus\kvector_2)(p) := \kvector_1(p)\gplus\kvector_2(p)$ for all $p\in P$, 
 \item We extend $\termMult$ from $\polys{\group}{\terms}\times\polys{\ints}{\terms}\fto\polys{\group}{\terms}$ to $\vectors{\group}\times\vectors{\ints}\fto\polys{\group}{\terms}$ by defining $(\kvector\termMult\kvector^\prime)(\theta) := \bigoplus_{p\in P} \kvector(p)\termMult\kvector^\prime(p)$ for all $\theta\in\terms$, 
 \item If $\sigma$ is an assignment, we define $\val{\kvector}\in\vectors{\group}$ by $\val{\kvector}(p) := \val{\kvector(p)}$ for all $p\in P$. 
\end{itemize}
\end{definition}

\noindent Let $\kvector_1\in\polys{\group}{\terms}^P$, $\kvector_2\in\polys{\ints}{\terms}^P$ and
$\delta$ be a substitution.
We observe: $\kvector_1\termMult(\subst[\delta]{\kvector_2}) =
\sum_{p\in P}\kvector_1(p)\termMult(\subst[\delta]{\kvector_2})(p) =
\sum_{p\in P}\kvector_1(p)\termMult(\subst[\delta]{\kvector_2(p)}) =
\sum_{p\in P}\subst[\delta]{(\kvector_1(p)\termMult\kvector_2(p))} = 
\subst[\delta]{(\kvector_1\termMult\kvector_2)}$.

\subsubsection{Algebraic Petri Nets}

An \emph{algebraic Petri net structure} consists of \emph{places} $P$ and \emph{transitions} $T$. A place $p\in P$ describes a token store, and a transition $t$ is given by two semi-positive $P$-vectors $\tminus$ and $\tplus$, describing token consumption and production, respectively.

\begin{definition}[Transition, algebraic Petri net structure]
Let $P\not=\emptyset$ be a set. A \emph{transition} $t = (\tminus,\tplus)$ over $P$ consists of two semi-positive simple $P$-vectors $\tminus,\tplus$ over $\ints$. We define the \emph{effect} $\tdelta\in\vectors{\ints}$ of $t$ by $\tdelta := -\tminus+\tplus$. Let $T$ be a set of transitions over $P$. Then, $(P,T)$ is an \emph{algebraic Petri net structure} (\apnstruc). We write $\pre{t}$ for $\{ p\in P \mid \tminus(p) > 0\}$.
\end{definition}

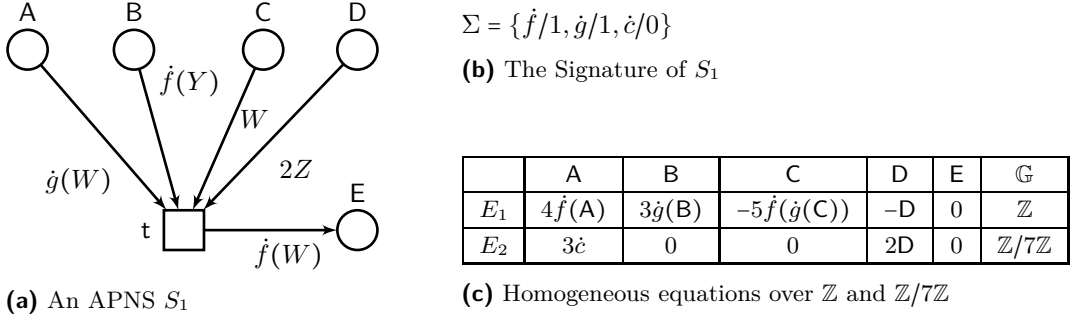
\begin{figure}[t]
\begin{subfigure}{0.42\textwidth}

\definecolor{C000000}{HTML}{000000}
\definecolor{CFFFFFF}{HTML}{FFFFFF}

\begin{tikzpicture}[inner sep=0pt, outer sep=0pt, node distance=0.5cm]
\node (n0) [minimum height=15.0pt,minimum width=15.0pt,at={(544.845431466667pt,-260.09051904pt)}, draw, line width=1.0pt, shape=circle] {};\node[above of=n0]{\bsptext{B}};
\node (n1) [minimum height=15.0pt,minimum width=15.0pt,at={(593.202279253333pt,-260.09051904pt)}, draw, line width=1.0pt, shape=circle] {};\node[above of=n1]{\bsptext{C}};
\node (n2) [minimum height=15.0pt,minimum width=15.0pt,at={(563.54851584pt,-328.15919104pt)}, fill=CFFFFFF, draw, line width=1.0pt, shape=rectangle] {};\node[left of=n2]  {\bsptext{t}};
\node (n3) [minimum height=15.0pt,minimum width=15.0pt,at={(505.046368426667pt,-260.09051904pt)}, draw, line width=1.0pt, shape=circle] {};\node[above of=n3]{\bsptext{A}};
\node (n4) [minimum height=15.0pt,minimum width=15.0pt,at={(628.45192704pt,-260.09051904pt)}, draw, line width=1.0pt, shape=circle] {};\node[above of=n4]{\bsptext{D}};
\node (n5) [minimum height=12.719467520000023pt,minimum width=23.289113600000064pt,at={(524.190925226667pt,-309.65450624pt)}, line width=1.0pt, shape=rectangle] {\bsptext{$\dot{g}(W)$}};
\node (n6) [minimum height=12.719467520000023pt,minimum width=23.289113600000064pt,at={(565.56958592pt,-271.9934192pt)}, line width=1.0pt, shape=rectangle] {\bsptext{$\dot{f}(Y)$}};
\node (n7) [minimum height=12.719467520000023pt,minimum width=23.289113600000064pt,at={(589.90727552pt,-286.11149632pt)}, line width=1.0pt, shape=rectangle] {\bsptext{$W$}};
\node (n8) [minimum height=12.719467520000023pt,minimum width=23.289113600000064pt,at={(605.54742144pt,-305.20406784pt)}, line width=1.0pt, shape=rectangle] {\bsptext{$2Z$}};
\node (n9) [minimum height=15.0pt,minimum width=15.0pt,at={(628.45192704pt,-328.15919104pt)}, draw, line width=1.0pt, shape=circle] {};\node[above of=n9]{\bsptext{E}};
\node (n10) [minimum height=12.719467520000023pt,minimum width=23.289113600000064pt,at={(602.55382144pt,-337.16612544pt)}, line width=1.0pt, shape=rectangle] {\bsptext{$\dot{f}(W)$}};

\path [-latex',line width=1.0pt,draw,](n0)--(n2);

\path [-latex',line width=1.0pt,draw,](n3)--(n2);

\path [-latex',line width=1.0pt,draw,](n4)--(n2);

\path [-latex',line width=1.0pt,draw,](n1)--(n2);

\path [-latex',line width=1.0pt,draw,](n2)--(n9);
\end{tikzpicture}

\caption{An \apnstruc $S_1$}
\end{subfigure}
\begin{subfigure}{0.57\textwidth}
\begin{subfigure}{\textwidth}
$\Sigma = \{ \dot{f}/1, \dot{g}/1, \dot{c}/0 \}$
\caption{The Signature of $S_1$}
\end{subfigure}
\begin{subfigure}{\textwidth}
\vspace{0.9cm}
%
\begin{tabularx}{\textwidth}{|X|c|c|c|c|c|c|}
\hline
& \bsptext{A} & \bsptext{B} & \bsptext{C} & \bsptext{D} & \bsptext{E} & $\group$ \\
\hline

$E_1$& $4\dot{f}(\bsptext{A})$ & $3\dot{g}(\bsptext{B})$ &  $- 5\dot{f}(\dot{g}(\bsptext{C}))$ &  $-\bsptext{D}$ & $0$ & $\ints$ \\
\hline
$E_2$& $3\dot{c}$ & $0$ & $0$ & $2\bsptext{D}$ & $0$ & $\ints/7\ints$ \\
\hline
\end{tabularx}
\caption{Homogeneous equations over $\ints$ and $\ints/7\ints$}
\end{subfigure}
\end{subfigure}
\caption{An \apnstruc $S_1$ with equations $E_1$ and $E_2$}
\label{fig:net-equations}
\end{figure}

\noindent \Fref{fig:net-equations} shows an example of an \apnstruc $S_1$ with transition $\bsptext{t}$, places 
\bsptext{A}, \bsptext{B}, \bsptext{C}, \bsptext{D} and \bsptext{E} and signature 
$\Sigma$ using two unary function symbols $\dot{f}$ and $\dot{g}$ and the 
constant $\dot{c}$. Transition $\bsptext{t}$ consists of
$\bsptext{t}^- = (\dot{g}(W) \quad \dot{f}(Y) \quad W \quad 2Z \quad 0 )$ and
$\bsptext{t}^+ = (0 \quad 0 \quad 0 \quad 0 \quad \dot{f}(W) )$.

A \emph{token} is a ground term, a \emph{marking} maps each place to a multiset 
of tokens: 
\begin{definition}[Marking]
Let $(P,T)$ be an \apnstruc. Let $\marking\in\vectors{\ints}$ be a semi-positive $P$-vector over $\ints$ with $\support{\marking}\subseteq\gterms$. Then, $\marking$ is a \emph{marking} of $(P,T)$. We write $\markings$ for the set of all markings of $(P,T)$.
\end{definition}

\noindent\emph{Algebraic Petri net semantics} are defined by the notion of a \emph{step} based on the effect of a transition, and the notion of a \emph{firing mode}: 
\begin{definition}[Step]
Let $(P,T)$ be an \apnstruc, $\marking,\marking^\prime\in\markings$, $t\in T$, and $\sigma$ be an assignment, such that $\marking\geq\val{\tminus}$ and $\marking^\prime = \marking + \val{\tdelta}$. Then, $\marking$ \emph{enables} transition $t$ in \emph{firing mode} $\sigma$, denoted by $\step{\marking}{t}{\sigma}{}$, and $(\marking,t,\sigma,\marking^\prime)$ is a step of $(P,T)$, denoted by $\step{\marking}{t}{\sigma}{\marking^\prime}$. 
\end{definition}
\noindent We remark that our definition of \emph{enabling} does not consider additional equality axioms; permitting such axioms is left for future work.  

An \emph{algebraic Petri net} \apn is an \apnstruc together with an \emph{initial marking}. Subsequent steps from the initial marking are \emph{runs}, the resulting markings are \emph{reachable}:  

\begin{definition}[Algebraic Petri net, run, reachable]
Let $(P,T)$ be an \apnstruc, and $\marking_0\in\markings$. Then, $(P,T,\marking_0)$ is an \emph{algebraic Petri net} (\apn). Let $\runldots{\marking_0}{t_1}{\sigma_1}{\marking_1}{\marking_{n-1}}{t_n}{\sigma_n}{\marking_n}$ be a sequence of steps. Then, $(t_1,\sigma_1)\ldots(t_n,\sigma_n)$ is a \emph{run} of $(P,T,\marking_0)$ and $\marking_n$ is \emph{reachable} in $(P,T,\marking_0)$.
\end{definition}

\subsection{Homogeneous Linear Equations of \apns}

A homogeneous \emph{(linear) $P$-equation} over a set $P$ of places has the form $\sum_{p\in P} \kcoeff_p\kterm_p = \gzero$, where $\kcoeff_p\in\group$ ($p\in P$) are elements of an abelian group $\group$ with $\gzero$ as neutral element and each $\kterm_p$ ($p\in P$) is a term. Formally, a homogeneous $P$-equation is given by a simple $P$-vector.

\begin{definition}[Homogeneous $P$-equation]
 Let $P$ be a set, $\group$ be an abelian group and $\kvector\in\vectors{\group}$ be simple.
 Then, $\kvector$ induces a \emph{homogeneous $P$-equation} over $\group$.
\end{definition}
\noindent\Fref{fig:net-equations} shows two equations $E_1$ and $E_2$. $E_1$ is over the group of integer $\ints$ and $E_2$ is over the group of integers modulo $7$, $\ints/7\ints$.
The table shows the simple $P$-vectors. For instance, $\kvector_1(\bsptext{A})\termMult\mathcal{X_{\bsptext{A}}}$ is the monomial
$\ml{\dot{f}(\bsptext{A})}{4}$.

A marking $\marking$ satisfies $E$ if replacing $\psym$ by $\marking$ yields an identity. A homogeneous $P$-equation is \emph{valid} in an \apn if it is satisfied by each reachable marking.

\begin{definition}[Satisfaction, validity]
Let $(P,T)$ be an \apnstruc, $\marking$ be a marking, $\group$ be an abelian group, and $E$ be a homogeneous $P$-equation over $\group$ given by the simple $P$-vector $\kvector\in\vectors{\group}$. If $\kvector\termMult\marking = \gzero$, then $\marking$ \emph{satisfies} $E$. If each reachable marking of $(P,T,\marking)$ satisfies $E$, then $E$ is \emph{valid} in $(P,T,\marking)$. 
\end{definition}
\noindent A homogeneous $P$-equation is \emph{stable} if satisfaction is preserved by all steps:

\begin{definition}[Stability]
Let $(P,T)$ be an \apnstruc, $t\in T$, $\group$ be an abelian group, and $E$ be a homogeneous $P$-equation over $\group$. Then, $E$ is \emph{$t$-stable} in $(P,T)$ iff for each step $\step{\marking}{t}{\sigma}{\marking^\prime}$ of $(P,T)$: If $\marking$ satisfies $E$, then $\marking^\prime$ satisfies $E$. 
\end{definition}

\noindent Stability together with satisfaction in the initial marking yields validity: 
\begin{lemma}
\label{lem:inductive-invariants}
 Let $(P,T,\marking)$ be an \apn, $\group$ be an abelian group, and $E$ be a homogeneous $P$-equation over $\group$ given by a simple $P$-vector $\kvector\in\vectors{\group}$. If $E$ is $t$-stable for each $t\in T$, and $\marking$ satisfies $E$, then $E$ is valid in $(P,T,\marking)$. 
\end{lemma}
A \emph{place invariant} $\kvector$ is a simple $P$-vector such that for each $t\in T$, we have $\kvector\termMult\tdelta = \gzero$. Then, the homogeneous equation induced by $\kvector$ is stable:
\begin{lemma}
\label{lem:invariants}
Let $(P,T)$ be an \apn,
$\group$ be an abelian group, and $E$ be a homogeneous $P$-equation over $\group$ given by a simple $P$-vector $\kvector\in\vectors{\group}$. Let $t\in T$ and $\kvector\termMult\tdelta = \gzero$. Then, $E$ is $t$-stable in $(P,T)$. 
\end{lemma}

\section{Contributions}
\label{sec:results}

We summarize our contributions in the form of two main theorems which we prove in the subsequent sections. Our first contribution is a proof that validity of a given $P$-equation in an \apn is undecidable. The proof can be found in \sref{sec:undecidability} and bases on a reduction of the halting problem of Minsky machines. 

\begin{theorem}
\label{thm:undecidable}
Let $(P,T,\marking)$ be an \apn and $E$ a homogeneous $P$-equation. Then, validity of 
$E$ in $(P,T,\marking)$ is undecidable.
\end{theorem}
\begin{proof}
Follows from Lemma~\ref{lem:minsky-halt} and Lemma~\ref{lem:step-minsky}. 
\end{proof}

\noindent
Our second contribution is a decidability proof for the stability of a homogeneous $P$-equation in an \apnstruc under the assumption that the coefficients stem from a cyclic group. Here, we develop a decidable, necessary and sufficient criterion, generalizing the invariant theorem (cf. \lref{lem:invariants}), in~\sref{sec:decide-stable}. 

\begin{theorem}
  \label{thm:decidable}
Let $(P,T)$ be an \apnstruc and $E$ be a homogeneous $P$-equation over a cyclic group, then
stability of $E$ in $(P,T)$ is decidable.
\end{theorem}
\begin{proof}
Follows from Lemma~\ref{lem:given-finite-spanning-set-decidable} and Lemma~\ref{lem:finite-spanning-set-computable}. 
\end{proof}


\section{Undecidability of Validity of Homogeneous Equations}
\label{sec:undecidability}

In this section, we give short description how to encode a
\emph{Minsky Machine}~\cite{Minsky1967} $M$ into an \apn $N_M$ using the Herbrand structure.
Then, the halting problem in the Minsky Machine reduces to validity of an equation.
This proof technique has been used before for Petri Nets, for example in \cite{ReynierS09}.
First, we recall the required notions of a Minsky machine, its states and its steps: 
 
\begin{definition}[Minsky machine]
A \emph{Minsky Machine} $M=(\mathcal{I}, \mathcal{R})$ consists of number of registers 
$\mathcal{R}\in\nats$ and a sequence $\mathcal{I}=I_1,\dots,I_n$ of 
instructions, where each instruction
$I_i \in \{ INC(r,z) \mid 1\leq r\leq\mathcal{R}, 1\leq z \leq n \} \cup
\{ JZ(r,z_1,z_2) \mid  1\leq r\leq\mathcal{R}, 1\leq z_1 \leq n-1, 1\leq z_2 \leq n-1 \}$
and $I_n = HALT$.

Every tuple $(\rho,\ell)\in\nats^{\mathcal{R}}\times\{1,\dots,n\}$ is a \emph{state} of $M$.
If $I_\ell = INC(r,z)$, then $(\rho,\ell)\rightarrow(\rho',z)$ is a step in $M$ with
$\rho'(r) = \rho(r) +1$ and $\rho'(q) = \rho(q)$ for all $q\not=r$.
If $I_\ell = JZ(r,z_1,z_2)$ and $\rho(r)>0$, then $(\rho,\ell)\rightarrow(\rho',z_1)$ is a step in $M$ with
$\rho'(r) = \rho(r) -1$ and $\rho'(q) = \rho(q)$ for all $q\not=r$.
If $I_\ell = JZ(r,z_1,z_2)$ $\rho(r)=0$, then $(\rho,\ell)\rightarrow(\rho,z_2)$ is a step.
We denote the reflexive transitive closure of $\rightarrow$ with $\rightarrow^*$.
\end{definition}
We recall that the halting problem for Minsky machines is undecidable:  
\begin{lemma}[\cite{Minsky1967}]
\label{lem:minsky-halt}
Let $M$ be a Minsky Machine.
It is undecidable, whether $M$ halts, i.e.
the following problem is undecidable: $\exists\rho\in\nats^\mathcal{R}$ such that
$(0,1) \rightarrow^* (\rho,n)$.
\end{lemma}
To reduce the halting problem, we encode a Minsky Machine into an \apnstruc.
\begin{definition}[Encoding of Minsky Machine]
Let $M$ be a Minsky Machine $M$,
then the \apnstruc $N_M$ encodes $M$, if:
\begin{itemize}
 \item The signature is $\Sigma_M = \{ \dot{f}/1, \dot{c}/0\}$,
 \item the set of places is $P = \{ p_r \mid 1\leq r \leq \mathcal{R}\} \cup \{ q_i 
\mid 1\leq i \leq n \}$,
\item for every $INC$-instruction $I_i$, let $t_i$ be the transition with the pattern 
shown in \fref{fig:inc-t},
\item and
for every $JZ$-instruction $I_i$ let $t_i$ and 
$t_i^\prime$ be the transitions following the pattern shown in \fref{fig:jz-t}.
\end{itemize}
\end{definition}

\begin{figure}
 \begin{subfigure}{0.45\textwidth}
  
\definecolor{C000000}{HTML}{000000}
\definecolor{CFFFFFF}{HTML}{FFFFFF}

\begin{tikzpicture}[inner sep=0pt, outer sep=0pt, node distance=0.5cm]
\node (n0) [minimum height=15.0pt,minimum width=15.0pt,at={(529.68088448pt,-371.50278016pt)}, draw, line width=1.0pt, shape=circle] {};\node[above of=n0]{\bsptext{$q_i$}};
\node (n1) [minimum height=15.0pt,minimum width=15.0pt,at={(630.18019456pt,-371.50278016pt)}, draw, line width=1.0pt, shape=circle] {};\node[above of=n1]{\bsptext{$q_z$}};
\node (n2) [minimum height=15.0pt,minimum width=15.0pt,at={(579.93053952pt,-418.14344064pt)}, draw, line width=1.0pt, shape=circle] {};\node[left of=n2]  {\bsptext{$p_r$}};
\node (n3) [minimum height=15.0pt,minimum width=15.0pt,at={(579.93053952pt,-371.50278016pt)}, fill=CFFFFFF, draw, line width=1.0pt, shape=rectangle] {};\node[above of=n3]{\bsptext{$t_i$}};
\node (n4) [minimum height=21.873536pt,minimum width=66.17587200000003pt,at={(602.79044288pt,-363.7800064pt)}, line width=1.0pt, shape=rectangle] {\bsptext{$\dot{c}$}};
\node (n5) [minimum height=21.873536pt,minimum width=66.17587200000003pt,at={(551.4502688pt,-362.77337344pt)}, line width=1.0pt, shape=rectangle] {\bsptext{$\dot{c}$}};
\node (n6) [minimum height=21.873536pt,minimum width=66.17587200000003pt,at={(604.59225856pt,-395.74460928pt)}, line width=1.0pt, shape=rectangle] {\bsptext{$\dot{f}(X)$}};
\node (n7) [minimum height=4.6909875200000215pt,minimum width=20.536481280000032pt,at={(562.28639296pt,-395.74460928pt)}, line width=1.0pt, shape=rectangle] {\bsptext{$X$}};

\path [-latex',line width=1.0pt,draw,](n2) edge  [bend left] (n3);

\path [-latex',line width=1.0pt,draw,](n3) edge  [bend left] (n2);

\path [-latex',line width=1.0pt,draw,](n0)--(n3);

\path [-latex',line width=1.0pt,draw,](n3)--(n1);
\end{tikzpicture}

  \caption{Encoding an instruction $I_i = INC(r,z)$}
  \label{fig:inc-t}
 \end{subfigure}
 \begin{subfigure}{0.55\textwidth}
  
\definecolor{CFFFFFF}{HTML}{FFFFFF}
\definecolor{C000000}{HTML}{000000}

\begin{tikzpicture}[inner sep=0pt, outer sep=0pt, node distance=0.5cm]
\node (n0) [minimum height=15.0pt,minimum width=15.0pt,at={(582.988670720001pt,-409.64398592pt)}, fill=CFFFFFF, draw, line width=1.0pt, shape=rectangle] {};\node[below of=n0]  {\bsptext{$t'_i$}};
\node (n1) [minimum height=15.0pt,minimum width=15.0pt,at={(507.3671872pt,-371.50278016pt)}, draw, line width=1.0pt, shape=circle] {};\node[below of=n1]  {\bsptext{$q_i$}};
\node (n2) [minimum height=15.0pt,minimum width=15.0pt,at={(633.415554560001pt,-335.10278784pt)}, draw, line width=1.0pt, shape=circle] {};\node[below of=n2]  {\bsptext{$q_{z_1}$}};
\node (n3) [minimum height=15.0pt,minimum width=15.0pt,at={(582.988670720001pt,-371.50278016pt)}, draw, line width=1.0pt, shape=circle] {};\node[left of=n3]  {\bsptext{$p_r$}};
\node (n4) [minimum height=15.0pt,minimum width=15.0pt,at={(582.988670720001pt,-335.10278784pt)}, fill=CFFFFFF, draw, line width=1.0pt, shape=rectangle] {};\node[above of=n4]{\bsptext{$t_i$}};
\node (n5) [minimum height=21.873536pt,minimum width=66.17587200000003pt,at={(608.576606720001pt,-328.23035904pt)}, line width=1.0pt, shape=rectangle] {\bsptext{$\dot{c}$}};
\node (n6) [minimum height=21.873536pt,minimum width=66.17587200000003pt,at={(528.40966656pt,-336.9482816pt)}, line width=1.0pt, shape=rectangle] {\bsptext{$\dot{c}$}};
\node (n7) [minimum height=21.873536pt,minimum width=66.17587200000003pt,at={(526.11903104pt,-407.90277248pt)}, line width=1.0pt, shape=rectangle] {\bsptext{$\dot{c}$}};
\node (n8) [minimum height=21.873536pt,minimum width=66.17587200000003pt,at={(606.02100352pt,-416.8398656pt)}, line width=1.0pt, shape=rectangle] {\bsptext{$\dot{c}$}};
\node (n9) [minimum height=21.873536pt,minimum width=66.17587200000003pt,at={(598.84582144pt,-390.57338304pt)}, line width=1.0pt, shape=rectangle] {};\node[below right =4.0pt and 0.6660610000000133pt of n9.north west,align=center,text width=64.84375,font=\fontfamily{phv}\fontsize{12}{13}\selectfont]  {\bsptext{$\dot{c}$}};
\node (n10) [minimum height=21.873536pt,minimum width=66.17587200000003pt,at={(565.818359040001pt,-393.72380032pt)}, line width=1.0pt, shape=rectangle] {};\node[below right =-0.09521399999999858pt and 0.6660610000000133pt of n10.north west,align=center,text width=64.84375,font=\fontfamily{phv}\fontsize{12}{13}\selectfont]  {\bsptext{$\dot{c}$}};
\node (n11) [minimum height=21.873536pt,minimum width=66.17587200000003pt,at={(598.84582144pt,-353.53955584pt)}, line width=1.0pt, shape=rectangle] {\bsptext{$X$}};
\node (n12) [minimum height=21.873536pt,minimum width=66.17587200000003pt,at={(559.414000640001pt,-356.97632384pt)}, line width=1.0pt, shape=rectangle] {};\node[below right =-0.09521399999999858pt and -6.939407749999987pt of n12.north west,align=center,text width=80.0546875,font=\fontfamily{phv}\fontsize{12}{13}\selectfont]  {\bsptext{$\dot{f}(X)$}};
\node (n13) [minimum height=15.0pt,minimum width=15.0pt,at={(633.415554560001pt,-409.64398592pt)}, draw, line width=1.0pt, shape=circle] {};\node[below of=n13]  {\bsptext{$q_{z_2}$}};

\path [-latex',line width=1.0pt,draw,](n3) edge  [bend right] (n0);

\path [-latex',line width=1.0pt,draw,](n0) edge  [bend right] (n3);

\path [-latex',line width=1.0pt,draw,](n3) edge  [bend left] (n4);

\path [-latex',line width=1.0pt,draw,](n4) edge  [bend left] (n3);

\path [-latex',line width=1.0pt,draw,](n1) edge  [bend left] (n4);

\path [-latex',line width=1.0pt,draw,](n1) edge  [bend right] (n0);

\path [-latex',line width=1.0pt,draw,](n4)--(n2);

\path [-latex',line width=1.0pt,draw,](n0)--(n13);
\end{tikzpicture}

  \caption{Encoding an instruction $I_i = JZ(r,z_1,z_2)$}
  \label{fig:jz-t}
 \end{subfigure}
 \caption{Encoding Minsky Machines into \apns}
\end{figure}
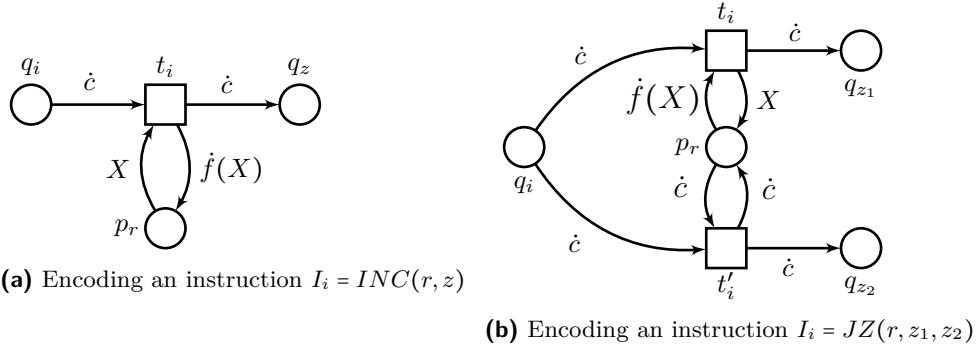

\begin{definition}
\label{def:marking-minsky}
 Let $(\rho,\ell)\in\nats^\mathcal{R}$ be a state of $M$. 
 For $x\in\nats$, we define $\theta_x\in\terms$ by 
 \begin{equation*}
  \theta_x := \begin{cases}
             \dot{c} & \text{ if }x=0\\
             \dot{f}(\theta_{x-1}) & \text{ otherwise}.
            \end{cases}
 \end{equation*}
Then, we define the marking $\marking^\rho_\ell \in \markings$ of $N_M$ as 
follows for $p\in P$ and $\theta\in\Theta$:
 \[\marking(p) :=
 \begin{cases}
  \ml{\dot{c}}{1} & \text{ if } p=q_\ell \\
  \ml{\theta_{\rho(r)}}{1} & \text{ if } p=p_r \\
  0 & \text{ otherwise }
 \end{cases} \] 
\end{definition}
Now, we can relate the steps of a Minsky Machine $M$ to the steps
of the encoding $N_M$.

\begin{lemma}
\label{lem:minsky-step}
 Let $(\rho,\ell),(\rho',\ell')$ be states of $M$ with $(\rho,\ell)\rightarrow(\rho',\ell')$. Then: 
 \begin{enumerate}
  \item There exists a step $\step{\marking^\rho_\ell}{t}{\sigma}{\marking^\prime}$ of $N_M$. 
  \item If $\step{\marking^\rho_\ell}{t}{\sigma}{\marking^\prime}$ is a step of $N_M$, then $\marking^\prime = \marking^{\rho^\prime}_{\ell^\prime}$. 
 \end{enumerate}
\end{lemma}
%
Finally, we reduce the halting problem for $M$ to the validity of the homogeneous $P$-equation $q_n = 0$ in $(N_M,\marking^1_0)$.
The $P$-vector over $\ints$ that induces the $P$-equation is zero for all places $p\in P\setminus\{q_n\}$ and $1$ for $q_n$.
Inductively applying \lref{lem:minsky-step} reduces reachability of the $HALT$ state
in $M$ to non-emptiness of the place $q_n$ in $(N_M,\marking^1_0)$ and thus to validity of $q_n = 0$.

\begin{lemma}
\label{lem:step-minsky}
 The equation $q_n = 0$ is valid in $(N_M,\marking^1_0)$ if and only if the Minsky Machine $M$ does not halt.
\end{lemma}

\section{Deciding Stability of Homogeneous Equations over Cyclic Groups}
\label{sec:decide-stable}

In this section, we show that stability of a homogeneous $P$-equation $E$ given by a simple $P$-vector $\kvector$ in an \apnstruc $N=(P,T)$ is decidable, 
if $\group$ is a cyclic group. 
To this end, we identify a decidable, necessary and sufficient condition for stability, which generalizes the necessary but not sufficient condition given by the classical invariant theorem (cf. \lref{lem:invariants}). 
We develop our condition based on the following lemma, which directly follows from applying additivity arguments to the definition of stability: 
\begin{lemma}
\label{lem:semantic-invariant}
 Let $t\in T$ be a transition. Then, the following statements are equivalent:
 \begin{enumerate}
  \item $E$ is $t$-stable. 
  \item For all steps $\step{\marking}{t}{\sigma}{\marking^\prime}$: If $\kvector\termMult\marking = \gzero$, then $\kvector\termMult\val{\tdelta} = \gzero$. 
 \end{enumerate}
\end{lemma}

\begin{figure}
\begin{center}

\definecolor{CE9E9E9}{HTML}{E9E9E9}
\definecolor{C000000}{HTML}{000000}
\begin{tikzpicture}[inner sep=0pt, outer sep=0pt, node distance=0.5cm]
\node (n0) [minimum height=17.359324123180798pt,minimum width=140.7558399999998pt,at={(1094.75424pt,-398.026044123181pt)}, fill=CE9E9E9, draw, line width=1.0pt, shape=rectangle] {\bsptext{Implementation}};
\node (n1) [minimum height=17.359324123180798pt,minimum width=140.7558399999998pt,at={(1094.75424pt,-450.518986184771pt)}, fill=CE9E9E9, draw, line width=1.0pt, shape=rectangle] {\bsptext{Satisfying Marking}};
\node (n2) [minimum height=17.359324123180798pt,minimum width=140.7558399999998pt,at={(1092.93312pt,-345.53310206159pt)}, fill=CE9E9E9, draw, line width=1.0pt, shape=rectangle] {\bsptext{Zero (\dref{def:zero})}};
\node (n3) [minimum height=17.359324123180798pt,minimum width=140.7558399999998pt,at={(1316.58976pt,-398.026044123181pt)}, fill=CE9E9E9, draw, line width=1.0pt, shape=rectangle] {\bsptext{Realization}};
\node (n4) [minimum height=17.359324123180798pt,minimum width=140.7558399999998pt,at={(1316.58976pt,-450.518986184771pt)}, fill=CE9E9E9, draw, line width=1.0pt, shape=rectangle] {\bsptext{Step from Satisfying Marking}};
\node (n5) [minimum height=17.359324123180798pt,minimum width=140.7558399999998pt,at={(1314.76864pt,-345.53310206159pt)}, fill=CE9E9E9, draw, line width=1.0pt, shape=rectangle] {\bsptext{Derivation}};
\node (n6) [minimum height=33.604480000000024pt,minimum width=97.82976000000008pt,at={(1122.33248pt,-424.272515153976pt)}, line width=1.0pt, shape=rectangle] {\bsptext{\lref{lem:markings-solutions}}};
\node (n7) [minimum height=33.604480000000024pt,minimum width=97.82976000000008pt,at={(1125.40448pt,-371.779573092385pt)}, line width=1.0pt, shape=rectangle] {\bsptext{\dref{def:implementation}}};
\node (n8) [minimum height=33.604480000000024pt,minimum width=97.82976000000008pt,at={(1203.85088pt,-337.410524123181pt)}, line width=1.0pt, shape=rectangle] {\bsptext{\dref{def:derivable}}};
\node (n9) [minimum height=33.604480000000024pt,minimum width=97.82976000000008pt,at={(1348.11488pt,-371.779573092385pt)}, line width=1.0pt, shape=rectangle] {\bsptext{\dref{def:realization}}};
\node (n10) [minimum height=33.604480000000024pt,minimum width=97.82976000000008pt,at={(1345.04288pt,-424.272515153976pt)}, line width=1.0pt, shape=rectangle] {\bsptext{\lref{lem:realization-enabling}}};

\path [-triangle 60,line width=1.0pt,draw,](n2)--(n5);

\path [-triangle 60,line width=1.0pt,draw,](n5)--(n3);

\path [triangle 60-triangle 60,line width=1.0pt,draw,](n3)--(n4);
\path (n3) ++(-3.999961679687885pt,-17.566796159835pt)node[align=center,text width=4.0,font=\fontfamily{phv}\fontsize{12}{13}\selectfont]{};

\path [-triangle 60,line width=1.0pt,draw,](n2)--(n0);

\path [triangle 60-triangle 60,line width=1.0pt,draw,](n0)--(n1);
\end{tikzpicture}

\end{center}
\caption{Overview of the proof of \tref{thm:decidable}}
\label{fig:proof-overview}
\end{figure}

\noindent \lref{lem:semantic-invariant} generalizes \lref{lem:invariants} in the sense that we can derive \lref{lem:invariants} from \lref{lem:semantic-invariant}, but not vice versa. However, the condition stated in \lref{lem:semantic-invariant} does not directly infer a decision procedure, because the set of steps $\step{\marking}{t}{\sigma}{\marking^\prime}$ with $\kvector\termMult\val{\tdelta} = \gzero$ is infinite, that is, one has to reason about infinitely many markings $\marking$ and firing modes $\sigma$. Our approach copes with this challenge by applying symbolic techniques, that is, we finitely characterize the infinite set of all such $\marking$ and $\sigma$ conveniently for computation. \fref{fig:proof-overview} summarizes the notions applied in our proof: We first symbolically describe the set of $E$-satisfying markings by means of \emph{\zeros} and their \emph{implementations}. Then, we \emph{derive} symbolically described firing modes from \zeros, and characterize stability by means of \emph{realizability}. 

In order to simplify notation, we fix for this section an \apnstruc $(P,T)$, an abelian group $\group$, and a homogeneous $P$-equation $E$ given by a simple vector $\kvector\in\vectors{\group}$. Moreover, we assume that for each $p\in P$, $\kvector(p)$ is the monomial $\ml{\kterm_p}{\kcoeff_p}$, that is, 
$\kcoeff_p = \kvector(p)(\kterm_p) \in \group$ is the coefficient of the only term $\kterm_p$ in $\support{\kvector(p)}$. 

Our first goal is to abstractly characterize infinite sets of $E$-satisfying markings by means of a \emph{\zero}. 
Intuitively, an $E$-satisfying marking assigns ``right number'' of a ``right kind of tokens'' to each place. 

\begin{definition}[\Zero]
\label{def:zero}
Let $\nu \ffrom P\fto \nats$ such that $\sum_{p\in P}\nu(p) \kcoeff_p = 0$.
If the unification problem $U = \{\kterm_p \doteq \kterm_{p^\prime} \mid 
p,p^\prime\in P, \kcoeff_p,\kcoeff_{p^\prime},\nu(p),\nu(p^\prime)\not= 0\}$ is 
solvable, 
$\nu$ is a \emph{\zero} of $E$, and we write $\underline{\nu}$ for the most general unification of $U$. 
\end{definition}
We observe that $0$ is always a zero. Furthermore, the sum of two 
\zeros $\nu_1,\nu_2$ yield again $\sum_{_p\in P}\left(\nu_1(p)+\nu_2(p)\right)=0$,
but the unification problem is not necessarily solvable.
However, a \zero may be the sum of other \zeros.

\newcommand{\pretablespace}{\vspace{0.5cm}}

\begin{figure}[t]
\begin{subfigure}{\textwidth}
\pretablespace
\begin{tabularx}{\textwidth}{|l|l|l|l|l|X|X|X|}
 \hline
  & \bsptext{A} & \bsptext{B} & \bsptext{C} & \bsptext{D} &
\Zero of $E_1$? & 
\Zero of $E_2$? & $\varrho(\nu_i)$ 
\\
 \hline
 $\nu_1$ & 0  & 1 & 0  & 3 &  yes &  no & $\dot{g}(\bsptext{B})$ \\
 $\nu_2$ & 5  & 0 & 4  & 0 &  yes &  no & $\dot{f}(\dot{g}(\bsptext{C}))$ \\
 $\nu_3$ & 0  & 2 & 0  & 6 &  yes &  no & $\dot{g}(\bsptext{B})$ \\
 $\nu_4$ & 1  & 1 & 1  & 2 &  no  &  yes & $\dot{c}$ \\
 $\nu_5$ & 2  & 0 & 0  & 4 &  no  &  yes & $\dot{c}$ \\
  \hline
\end{tabularx}
\caption{Zeros $\nu_1,\ldots,\nu_5\in\nats^P$}
\label{fig:zeros-solutions}
\end{subfigure}
\begin{subfigure}{\textwidth}
\pretablespace
\begin{tabularx}{\textwidth}{|l|l|l|l|l|X|X|X|}
\hline
& $\bsptext{A}$ & $\bsptext{B}$ & $\bsptext{C}$ & $\bsptext{D}$ &
impl.\ $\nu_1$ for $E_1$? & impl.\ $\nu_2$ for $E_1$? & impl.\ $\nu_5$ for $E_2$? \\
\hline
$\marking_1$ & $0$ & $\dot{c}$ & $0$ & $3\dot{g}(\dot{c})$ &
yes & no & no \\
$\marking_2$ & $0$ & $2\dot{f}(\dot{c})$ & $0$ & $6\dot{g}(\dot{f}(\dot{c}))$ &
yes & no & no \\
$\marking_3$ & $5\dot{g}(\dot{c})$ & $0$ & $4\dot{c}$ & $0$ &
no & yes & no \\
$\marking_4$ & $2\dot{g}(\dot{c})$ & $0$ & $0$ & $4\dot{c}$ &
no & no & yes \\
\hline
\end{tabularx}
\caption{Implementations of \zeros $\nu_1$ (w.r.t. $E_1$), $\nu_2$ (w.r.t. $E_1$) and $\nu_5$ (w.r.t. $E_2$)}
\label{fig:implementations}
\end{subfigure}
 \caption{Examples for \zeros, realizations, and implementations}
\end{figure}

\Fref{fig:zeros-solutions} shows some examples for \zeros using the net structure and
equations shown in \fref{fig:net-equations}. In this section, we ignore the place $\bsptext{E}$,
as it is irrelevant for enabling $\bsptext{t}$.
$\nu_1$ is a \zero of $E_1$ as $3-3=0$, and $\dot{g}(\bsptext{B}) 
\doteq \bsptext{D}$ can be unified with $D\mapsto\dot{g}(\bsptext{B})$.
$\nu_2$ is a \zero of $E_1$ as $20-20=0$ and $\bsptext{A}\mapsto\dot{g}(\bsptext{C})$ unifies
$\dot{f}(\bsptext{A}) \doteq \dot{f}(\dot{g}(\bsptext{C}))$.
For $\nu_4$ and $E_1$ we have $4+3-5-2=0$, but it is not a \zero of $E_1$ as 
$\dot{f}(\bsptext{A}) \doteq \dot{g}(\bsptext{B})$ cannot be unified.
$\nu_5$ is not a \zero for $E_1$ as $8-4\not=0$.
Regarding $E_2$,
$\nu_1$ and $\nu_2$ aren't \zeros as $6 \not \equiv_7 0$ and $15 \not \equiv_7 0$.
$\nu_4$ is a zero for $E_2$ as $3+4\equiv_7 0$ and $\bsptext{D}\mapsto\dot{c}$ unifies
$\dot{c}\doteq\bsptext{D}$.
Finally, $\nu_5$ is also a \zero of $E_2$, as $6+8\equiv_7=0$ and as for $\nu_4$ the 
unification problem is solvable as for $\nu_4$.

Because $\underline{\nu}$ is a unifier, applying $\underline{\nu}$ to $\kterm_p$ 
yields the same result for every $p\in P$ satisfying $\kcoeff_p\not=\gzero$ and 
$\nu(p)\not= 0$ . 
\begin{lemma}
Let $\nu$ be a \zero. The set $\{\subst[\underline{\nu}]{\kterm_p}\mid p\in P, \kcoeff_p \not = \gzero, \nu(p) \not = 0\}$ is singleton. 
\end{lemma}

\begin{definition}[Result of the unification]
We define $\varrho(\nu)$ by $\{\varrho(\nu)\} = \{\subst[\underline{\nu}]{\kterm_p}\mid p\in P, \kcoeff_p \not = \gzero, \nu(p) \not = 0\}$
\end{definition}

\noindent Intuitively, an \emph{implementation} of a \zero $\nu$ is a marking which satisfies $E$ ``in the same way'' as $\nu$. Formally, we define this based on an assignment transforming the result of the unification to a marking. 

\begin{definition}[Implementation of a \zero]
\label{def:implementation}
 Let $\marking\in\markings$ be a marking and $\nu$ be a \zero for $E$.
 Then, $\marking$ \emph{implements} $\nu$, or: $\marking$ is an \emph{implementation} of $\nu$, if for all $p\in P$ with $\nu(p)\not=0$ and $\gamma_p\not=\gzero$: 
 \begin{enumerate}
  \item $\nu(p) = \sum_{\theta\in\support{\marking(p)}}\marking(p)(\theta)$, and
  \item there exists an assignment $\sigma$, such that
  $ \{\val{\varrho(\nu)}\} = \support{\kvector(p)\termMult\marking(p)}$.
 \end{enumerate}
\end{definition}

\noindent As an example, in \fref{fig:implementations}, the marking $\marking_1$ 
implements $\nu_1$ for $E_1$ as for assignment $\sigma_1$ with $\sigma_1(\bsptext{B})=\dot{c}$
we have $\val[\sigma_1]{B} = \dot{c} = \bsptext{D}\termMult\dot{g}(\dot{c})$. 
$\marking_2$ implements $\nu_1$ for $E_1$, because for assignment $\sigma_2$
with $\sigma_2(\bsptext{B})=\dot{f}(\dot{c})$, we have
$\val[\sigma_2]{\bsptext{B}} = 
\bsptext{D}\termMult\dot{f}(\dot{g}(\dot{c}))$.
$\marking_3$ implements $\nu_2$ for $E_1$, because for assignment $\sigma_3$
with $\sigma_3(\bsptext{C})=\dot{c}$, we have
$\val[\sigma_2]{\dot{f}(\dot{g}(\bsptext{C}))} = \dot{f}(\dot{g}(\bsptext{C})) =
\dot{f}(A)\termMult\dot{g}(\dot{c}) = \dot{f}(\dot{g}(\bsptext{C}))\termMult\dot{c}$.
Moreover, $\marking_4$ implements $\nu_5$ for $E_2$ as for assignment $\sigma_4$ with
$\sigma_4(\bsptext{D})=\dot{c}$ we have
$\val[\sigma_4]{\dot{c}} = \dot{c}\termMult\dot{g}(\dot{c}) = \bsptext{D}\termMult\dot{c}$.

Next, we show that the set of all \zeros exactly characterizes the set of all $E$-satisfying markings:
For every term $\omega$ used by an $E$-satisfying marking $\marking$ we can identify an implementation $\marking_\omega$ of a \zero.
Because the set of $E$-satisfying markings is closed under addition, the converse also holds.

\begin{lemma}
\label{lem:markings-solutions}
Let $\marking$ be a marking,
the following are equivalent: 
\begin{enumerate}
 \item $\kvector\termMult\marking = \gzero$.
 \item There exist zeros $\nu_1,\ldots,\nu_n$ of $E$, and markings $\marking_1,\ldots,\marking_n$,
 such that:
 $\marking = \sum_{1\leq i\leq n}\marking_i$ and $\marking_i$ implements $\nu_i$ for all $i=1,\dots,n$.
\end{enumerate}
\end{lemma}

\noindent Our next goal is to abstractly describe sets of firing modes \emph{derivable} from a set of \zeros. Formally, we describe such a set of derived firing modes by a substitution, abstractly describing a way of enabling a transition. 

\begin{definition}[Derivable]
\label{def:derivable}
Let $t\in T$. 
 Let $\spset$ be a set of \zeros.
 For every $q\in \pre{t}$ let $X_q \in \varset$ be a fresh variable,
 such that $X_q$ does not occur in $E$ or $t$ and $X_q = X_{q'}$ implies $q=q'$.
 Let $\nu_q\in S$ be a \zero with $\nu_q(q) \geq 1$.
 Let $U = \left\{ \varrho(\nu_q)\termMult X_q \doteq 
 \kterm_q\termMult\theta_{q,t} \mid q\in \pre{t} \right\}$,
 where $\{\theta_{q,t}\} = \support{\tminus(q)}$.
 Let $U$ be solvable by most general unification $\delta$. 
 Then, $\delta$ is \emph{derivable} from $\spset$.
\end{definition}

\begin{figure}[t]
\begin{tabularx}{\textwidth}{|l|l|l|l|X|X|}
\hline
 & \bsptext{W} & \bsptext{Y} & \bsptext{Z} & Derivable from some $E_j$? & $\kvector_j\termMult(\subst[\delta_i]{\tdelta})$ \\
 \hline
 $\delta_1$ & $X_{\bsptext{C}}$ & $X_{\bsptext{B}}$ & $\dot{g}(X_{\bsptext{B}})$ & yes, for $j= 1$ & $-\dot{f}(\dot{g}(X_{\bsptext{C}})) + \dot{g}(X_{\bsptext{B}})$ \hfill ($j= 1$)\\
 $\delta_2$ & $X_{\bsptext{C}}$ & $X_{\bsptext{B}}$ & $\dot{c}$ & yes, for $j = 2$ & $0$ \hfill ($j= 2$)\\
 \hline
\multicolumn{6}{l}{}\\
 \hline
 & \bsptext{W} & \bsptext{Y} & \bsptext{Z} & Realization of & $\kvector_1\termMult\val[\sigma_1]{\tdelta}$ \\
 \hline
 $\sigma_1$ & $\dot{c}$ & $\dot{c}$ & $\dot{g}(\dot{c})$ & $\delta_1$ & $-\dot{f}(\dot{g}(\dot{c})) + \dot{g}(\dot{c})$ \\
 \hline
\end{tabularx}
 \caption{Derivable substitutions $\delta_1$ and $\delta_2$, and a realization $\sigma_1$ of $\delta_1$}
 \label{fig:derivable-realizations}
\end{figure}

\noindent In the example of \fref{fig:derivable-realizations},
we can derive $\delta_1$ for $E_1$ with
$\nu_{\bsptext{A}} = \nu_{\bsptext{C}} = \nu_4$ and $\nu_{\bsptext{B}} = \nu_{\bsptext{D}} = \nu_1$.
For $E_2$, we can derive $\delta_2$ with
$\nu_{\bsptext{A}} = \nu_{\bsptext{B}} = \nu_{\bsptext{C}} = \nu_{\bsptext{D}} = \nu_5$.

A \emph{realization} is an assignment which refines a derivable substitution: 

\begin{definition}[Realization]
\label{def:realization}
 Let $\spset$ be a set of \zeros and $\delta$ be derivable from $\spset$.
 Then, $\sigma$ is a \emph{realization} of $\delta$, if there exists an assignment $\sigma'$ with
 $\sigma(X) = \val[\sigma']{\delta(X)}$ for all $X\in\varset$.
\end{definition}

\noindent
The assignment $\sigma_1$ shown in \fref{fig:derivable-realizations} is a realization of $\delta_1$.
The assignment $\sigma$ with $\sigma(X_{\bsptext{C}})=\sigma(X_{\bsptext{B}})=\dot{c}$
gives $\sigma_1(\bsptext{A}) = \val{X_{\bsptext{A}}} = \dot{c}$, $\sigma_1(\bsptext{B}) = \val{X_{\bsptext{B}}} = \dot{c}$ and
$\sigma_1(\bsptext{C}) = \val{\dot{g}(X_{\bsptext{B}})} = \dot{g}(\dot{c})$.

Next, we show that the derived substitutions from the set of all \zeros exactly characterize the set of $E$-satisfying, $t$-enabling markings: If an $E$-satisfying marking $\marking$ enables $t$ in firing mode $\sigma$, then $\sigma$ is a realization of some derivable substitution, and vice versa: 

\begin{lemma}
\label{lem:realization-enabling}
 Let $\spset$ be the set of all \zeros and $\sigma$ be an assignment.
 Then, the following two statements are equivalent:
 \begin{enumerate}
  \item There exists a marking $\marking$ with:
  $\marking\geq\val{\tminus}$ and $\kvector\termMult\marking=\gzero$.
  \item There exists a $\delta$ that is derivable from $\spset$ 
    and $\sigma$ is a realization of $\delta$.
 \end{enumerate}
\end{lemma}

\noindent A derivable substitution $\delta$ generally has infinitely many realizations.
We show that the choice of the realization does not matter for deciding stability.

\begin{lemma}
\label{lem:delta-zero-all-zero}
Let $\spset$ be a set of \zeros and $\delta$ be derivable from $\spset$.
Then, the following two statements are equivalent:
\begin{enumerate}
 \item $\kvector\termMult(\subst[\delta]{\tdelta}) = 0$
 \item $\kvector\termMult\val{\tdelta} = 0$ for all $\sigma$ that are realizations of $\delta$.
\end{enumerate}
\end{lemma}

\noindent Our proof of ``2.$\Rightarrow$1.''\,utilizes the existence of a realization $\sigma$ preserving the distinctness of terms in  $\kvector\termMult\tdelta$, that is, if two terms $\theta_1,\theta_2$ occur in $\kvector\termMult\tdelta$ with $\subst[\delta]{\theta_1}\ne\subst[\delta]{\theta_2}$, then $\val{\theta_1}\not=\val{\theta_2}$. 


Now, we prove that $t$-stability can be characterized by the set of all derivable substitutions:

\begin{lemma}
\label{lem:homo-stable-charac}
Let $\spset$ be the set of all \zeros.
 The following are equivalent:
 \begin{enumerate}
  \item $E$ is $t$-stable. 
  \item For all $\delta$ derivable from $\spset$ holds:
  $\kvector\termMult(\subst[\delta]{\tdelta}) = 0$.
 \end{enumerate}
\end{lemma}

\noindent In the example shown in \fref{fig:net-equations}, $E_1$ is not stable.
Consider the marking $\marking_5:=\marking_1+\marking_2+\marking_3$. There, $\bsptext{t}$ is enabled.
But, for the firing mode $\sigma_1$, we have $\kvector_1\termMult\sigma_1 \not=0$.
On the other hand, $E_2$ is stable, although we have $\kvector_2\termMult\tdelta \not = 0$.

The following lemma proves a closure property for the derived substitutions: If one combines \zeros from a set $\spset$ to a new \zero $\nu$, then for every realizable substitution derivable from $\spset\cup\{\nu\}$, there exists a realizable substitution derivable from $\spset$.

\begin{lemma}
\label{lem:finite-derivable}
 Let $\spset$ be a set of \zeros and $\nu\not\in\spset$
 with $\nu = \sum_{i=1}^n\nu_i$ where $\nu_i\in\spset$. 
 Let $\delta$ be derivable from $\spset \cup \{\nu\}$
 and $\sigma$ be assignments that realizes $\delta$.
 Then, there exists $\delta'$ such that:
 $\delta'$ is derivable from $\spset$ and $\sigma$ realizes $\delta'$.
\end{lemma}

\noindent We observe that we can only derive finite sets of substitutions from finite sets of \zeros. 

\begin{lemma}
\label{lem:finitely-many-derivable}
Let $\spset$ be a finite set of \zeros.
The set $\{ \delta\ffrom\varset\fto\Theta \,\mid\, \delta$ is derivable from $S \}$ is finite and computable.
\end{lemma}

\noindent Our next goal is to combine \lref{lem:finite-derivable} and \lref{lem:finitely-many-derivable}. To this end, we first define the notion of a \emph{spanning set} of zeros: A set capable of generating all zeros by means of addition. 

\begin{definition}[Spanning Set]
Let $S$ be a set of \zeros of $E$, such that for each \zero $\nu$ of $E$, there exist $\nu_1,\ldots,\nu_n\in S$, with $\nu(p) = \sum_{i = 1}^n \nu_i(p)$ for all $p\in P$. Then, $S$ is a \emph{spanning set} (of zeros) of $E$.
\end{definition}

\noindent Now, we show that given a \emph{finite} spanning set of \zeros, we can decide $t$-stability. 

\begin{lemma}
\label{lem:given-finite-spanning-set-decidable}
Given a finite spanning set $\spset$ of \zeros, $t$-stability of $E$ is decidable.
\end{lemma}

\begin{proof}
 By \lref{lem:finite-derivable}, for every $\delta$ that is derivable from the set of
 \zeros, there exists a $\delta^\prime$ derivable from $\spset$. By \lref{lem:finitely-many-derivable}, the set of all these $\delta^\prime$ is finite and computable.
 By \lref{lem:homo-stable-charac}, $E$ is stable if and only if for every
 $\delta'$ we have $\kvector\termMult\subst[\delta']{\tdelta} = 0$, which is computable.
\end{proof}

\noindent The last step in our proof of \tref{thm:decidable} is showing that a finite spanning set of zeros can be computed if $\group$ is cyclic. For infinite cyclic groups, we apply that there exists a computable isomorphism into the integers. As a prerequisite, we observe that every spanning set contains every \emph{indecomposable \zero},~i.e.,~a \zero which cannot be written as a sum of other \zeros. For example, consider the \zeros $\nu_1$, $\nu_2$ and $\nu_3$ from \fref{fig:zeros-solutions}: $\nu_1$ and $\nu_2$ are indecomposable, but $\nu_3 = \nu_1+\nu_1$ is not. 
Thus, we show that there exists an upper bound for the coefficients of indecomposable \zeros. To this end, we first show an auxiliary lemma, based on the maximum coefficient $\overline{\gamma}$, and the absolute value $\underline{\gamma}$ of the minimal coefficient in $\gamma$. In the example equation $E_1$ from \fref{fig:net-equations}, we have $\overline{\gamma} = 4$ and $\underline{\gamma} = 5$. Intuitively, if the maximum constituent in a zero $\nu$ over places with negative (resp. positive) coefficients is less than $\overline{\gamma}$ (resp. $\underline{\gamma}$), then the sum of the constituents in $\nu$ is bounded by $2|P|\overline{\gamma}\underline{\gamma}$. For $E_1$, the upper bound is $2 \cdot 5\cdot 4\cdot 5 = 200$.  
\begin{lemma}
\label{lem:zero-bound}
 Let $\nu\in \nats^P$. Let $\eta\in\ints^P$ be mixed with $\sum_{p\in P}\nu(p)\cdot\eta(p)=0$.
 Let $\overline{\eta}:=\max\{ \eta(p) \mid p\in P\}$
 and $\underline{\eta}:=\max\{ |\eta(p)| \mid \eta(p)<0, p\in P\}$ with:
 \begin{enumerate}
  \item $\max\{ \nu(p) \mid \eta(p) < 0, p\in P \} < \overline{\eta}$
  \item or $\max\{ \nu(p) \mid \eta(p) > 0, p\in P \} < \underline{\eta}$.
 \end{enumerate}
Then, $\sum_{p\in P}\nu(p) < 2|P|\overline{\eta}\underline{\eta}$
\end{lemma}
\noindent Finally, we show the computability of a finite spanning set of \zeros. To this end, we utilize \lref{lem:zero-bound} to show that the sum of constituents of each indecomposable \zero is bounded by $2|P|\overline{\gamma}\underline{\gamma}$: We assume a zero $\nu$ with $\sum_{p\in P}\nu(p)\geq 2|P|\overline{\gamma}\underline{\gamma}$, and show that 
$\nu$ decomposes into two zeros $\hat{\nu}$ and $\nu - \hat{\nu}$.
Thus, extracting all \zeros from the finite set of all $\nu\in\nats^P$ with $\sum_{p\in P}\nu(p) < 2|P|\overline{\gamma}\underline{\gamma}$ yields a set of \zeros containing all indecomposable \zeros, and hence a finite spanning set. 

\begin{lemma}
\label{lem:finite-spanning-set-computable}
If $\group$ is cyclic, a finite spanning set $\spset$ of \zeros is computable.
\end{lemma}

\begin{proof}
Assume $\kvector$ is semi-positive or semi-negative, then $0$ is the only \zero.
In the following, we assume $\kvector$ to have mixed coefficients. We distinguish the cases
whether $\group$ is finite or infinite.
\begin{itemize}
 \item First case: $\group$ is infinite.
 As $\group$ is cyclic, there exists a computable isomorphism to $\ints$
 (see for instance \cite{Whitelaw1978}).
Thus, we assume w.l.o.g that $\group=\ints$.
Let $\overline{\gamma}:=\max\{ \gamma(p) \mid p\in P\}$
and $\underline{\gamma}:=\max\{ |\gamma(p)| \mid \gamma(p)<0, p\in P\}$.
Let $\nu$ be a \zero with $\sum_{p\in P}\nu(p)>2|P|\underline{\gamma}\overline{\gamma}$ (*).
We show that that then, there exist $\underline{p},\overline{p}\in P$ with: 
$\gamma_{\overline{p}}>0 \wedge
\gamma_{\underline{p}}<0 \wedge
\nu(\overline{p})\geq|\gamma_{\underline{p}}| \wedge
\nu(\underline{p})\geq\gamma_{\overline{p}}$.
Assume the opposite: Then,
  $\max\{ \nu(p) \mid \gamma_p < 0, p\in P \} < \overline{\gamma}$
  or $\max\{ \nu(p) \mid \gamma_p > 0, p\in P \} < \underline{\gamma}$.
By \lref{lem:zero-bound}, then  $\sum_{p\in P}\nu(p) < 2|P|\overline{\gamma}\underline{\gamma}$,
which contradicts (*).

 Now, let $\hat{\nu}\ffrom P\fto\nats$ with:
 \[\hat{\nu}(p) =
 \begin{cases}
  |\gamma_{\underline{p}}| & \text{ if } p = \overline{p} \\
  \gamma_{\overline{p}} & \text{ if } p = \underline{p} \\
  0 & \text{ otherwise }
 \end{cases}\]
 By definition, we have $\hat{\nu} \leq \nu$, moreover
 as $\sum_{p\in P}\hat{\nu}(p) \leq \underline{\gamma}\overline{\gamma} 
<\sum_{p\in P}\nu(p)$, we have $\hat{\nu} < \nu$.
 Let $\nu' = \nu - \hat{\nu}$. Then, $\nu'\ffrom P \fto \nats$ and $\nu' > 0$.
 
 Now we show that $\hat{\nu}$ and $\nu'$ are zeros.
 For $\hat{\nu}$ we have
 $\sum_{p\in P}\nu(p)=|\gamma_{\underline{p}}|\gamma_{\overline{p}}+\gamma_{\underline{p}}\gamma_{\overline{p}}=
 -\gamma_{\underline{p}}\gamma_{\overline{p}}+\gamma_{\underline{p}}\gamma_{\overline{p}}=0$
 and accordingly $0=\sum{p\in P}\nu(p)=\sum{p_\in P}\hat{\nu}(p) + \sum_{p\in P}\nu'(p)=
 0 +\sum_{p\in P}\nu'(p)$.
 It remains to show that the unification problems of $\hat{\nu}$ and $\nu'$
 are solvable.
 We observe $\hat{\nu}\leq\nu$ ($\nu'\leq\nu$) implies that unification problem 
 of $\hat{\nu}$ ($\nu'$) is a subset of the unification problem of $\nu$.
 Thus, $\nu$ is a sum of the zeros $\nu'$ and $\hat{\nu}$.
 
 Now, we see that $\sum_{p\in P}\hat{\nu}(p) < 2|P|\overline{\gamma}\underline{\gamma}$.
 Assume additionally $\sum_{p\in P}\nu'(p) \leq 2|P|\overline{\gamma}\underline{\gamma}$,
 then we can continue.
 Otherwise, if $\sum_{p\in P}\nu'(p) > 2|P|\overline{\gamma}\underline{\gamma}$,
 we can apply induction, as $\nu' < \nu$,
 Hence, $\nu$ is the sum of other \zeros $\nu_1,\dots,\nu_n$, where for each $1\leq i \leq n$:
 $\sum_{p\in P}\nu_i(p) \leq 2|P|\underline{\gamma}\overline{\gamma}$.
 Finally, $\{\nu \in \nats^p \mid \sum_{p\in 
P}|\nu(p)|\leq|2P|\overline{\gamma}\underline{\gamma} \text{ and $\nu$ is 
zero}\}$ is finite, spanning and computable.
 \item Second Case: Let $\group$ be finite with order 
$o\in\nats\setminus\{0\}$. As $\group$ is cyclic, there exists the generator $e\in\group$.
Let $g\in\group$.
Then, it holds that 
$g + oe = g$. Thus, for every $\nu\ffrom P\fto \nats$, and $p\in P$ with $\nu(p)> o$, we 
have $\nu(p)\gamma_p = (\nu(p)-o)\gamma_p$. Hence, for every \zero $\nu$ we can 
find a \zero $\nu'$ with $\nu'(p)\leq o$ and $\sum_{p\in 
P}\gamma_p\nu(p)=\sum{p\in P}\gamma_p\nu'(p)$. Therefore, $\{\nu \in \nats^p 
\mid \nu(p)\leq o \text{ and $\nu$ is \zero}\}$ is finite, spanning and 
computable.\qed
\end{itemize}\fixqed
\end{proof}
%

\section{Related Work}
\label{sec:related}
APNs or similar ``high level net''-formalisms are an established, expressive modeling language for distributed systems\cite{Reisig2013,Reisig1997}. 
Moreover, tools for Colored Petri Nets support simulation and (partial) verification \cite{Jensen2009,Jensen2015}.
The idea to prove stable properties in Petri nets that use distinguishable tokens
has been pursued at least since the early 80s~\cite{Genrich1981}.
Ever since, the class of invariants became a substantial part of Petri Net 
analysis \cite{Memmi1986,Reisig1997,Reisig2013}.
Other stable properties for Algebraic Petri Nets have been studied in the context of
Traps/Co-Traps \cite{Karsten1997}. In elementary Petri Nets (P/T-Nets), stable properties such as traps and co-traps
have been studied \cite{Reisig2013} and been shown as useful for verification~\cite{Reisig2013,Esparza2014}.
Compared to this, the number of publications regarding stable properties in \apns is comparatively small.
In the last years, Petri Net variants with distinguishable tokens gained more 
attention to model data in distributed systems and applying analytic methods such as \cite{Esparza16,Hofman2016,RosaVelardo2011}.

The concept of stability has been used in other areas of research; the most similar maybe being
abstract interpretation as a technique for verification of iterative programs~\cite{Bouajjani2010}. 
In the context of data-aware business processes, stability has been used
in a similar context, following a graph-oriented approach focusing
on data modeling \cite{nutt2016}.

\section{Concluding Remarks}
\label{sec:conclusion}
Throughout this paper, we applied three restrictions: First, we only considered the interpretation of terms in the Herbrand structure, second, we only considered homogeneous $P$-equations, and third, we required for the decidability proof that the group of coefficients is cyclic. 

If one chooses another structure for the interpretation of terms than the Herbrand structure, one can observe that validity and stability are preserved in one direction: If a $P$-equation is valid (stable) w.r.t. the Herbrand structure, then it is valid (stable) w.r.t. every generated structure. Because the Herbrand structure is a specific structure, the undecidability result (\tref{thm:undecidable}) could be generalized by allowing an arbitrary, but not fixed, structure. For the decidability result (\tref{thm:decidable}), we observe that we can use our decision procedure as a sufficient but not necessary criterion for an arbitrary fixed structure.

The restriction to homogeneous $P$-equations yields that satisfying markings are closed under addition, which is not the case if one allowed arbitrary constants on the right hand side. Here, our approach of finding a finite spanning set symbolically describing all satisfying markings does not work. The main challenge for generalizing our approach is that markings have natural numbers as coefficients (in contrast to integers).

For our decidability result, we require that the coefficients stem from a cyclic group. Here, we explicitly exploit in the proofs that there exist a distinct generator element, and an isomorphism to the integers, or the integers modulo some natural number $n$. 



\bibliography{paper}

\appendix
\section{Missing proofs}
\label{sec:proofs}

\renewcommand{\lref}[1]{Lemma~\ref{#1} (page~\pageref{#1})}
\newcommand{\ptit}[1]{Proof of~\lref{lem:#1}}

\begin{proof}[\ptit{inductive-invariants}]
Let $\marking$ be a reachable marking.
Then, there exists a sequence of steps
$\runldots{\marking_0}{t_1}{\sigma_1}{\marking_1}{\marking_{n-1}}{t_n}{\sigma_n}{\marking_n}$
with $\marking_n = \marking$.
If $n=0$, then $\marking=\marking_0$ and thus $\kvector\termMult\marking=\gzero$.
Otherwise, by induction we have $\kvector\termMult\marking_{n-1}=\gzero$
and as $E$ is $t_n$-stable, satisfaction is preserved along $\step{\marking_{n-1}}{t_n}{\sigma_n}{\marking_n}$
and thus $\kvector\termMult\marking_n=\gzero$ and hence $\kvector\termMult\marking=\gzero$.
\end{proof}

\begin{proof}[\ptit{invariants}]
Let $\step{\marking}{t}{\sigma}{\marking'}$ be a step with $\kvector\termMult\marking=\gzero$.
Then, $\kvector\termMult\val{\tdelta}=\val{\kvector\termMult\tdelta}=\val{\gzero}=\gzero$.
We conclude that $\kvector\termMult\marking' = \kvector\termMult(\marking + \val{\tdelta}) =
\kvector\termMult\marking + \kvector\termMult\val{\tdelta} = \gzero + \gzero = \gzero$.
Hence, $\marking'$ satisfies $E$ and $E$ is $t$-stable in $(P,T)$.
\end{proof}

\begin{proof}[\ptit{minsky-step}]
By \dref{def:marking-minsky}, $\marking^\rho_\ell(p_r)$ is the monomial $\ml{\theta_{\rho(r)})}{1}$, and $\marking(q_k)(\dot{c}) = 0$ for all $k\not=\ell$. We distinguish based on the type of $I_i$. 
\begin{itemize}
 \item Let $I_i = INC(r,z)$. We observe $\tminus_i(p_r)$ is the monomial $\ml{X}{1}$, and for all $1\leq s\leq \mathcal{R}$, $s\not=r$: $\tminus_i(p_s) = 0$. Furthermore, $\tminus_i(q_i)$ is the monomial $\ml{\dot{c}}{1}$ and $\tminus_i(q_j) = 0$ for all $1\leq j\leq n$, $j\not= i$. Hence, $t_i$ is enabled in firing mode $\sigma$ with $\sigma(X) = \theta_{\rho(r)}$: $\marking^\ell_\rho(p_r) = \theta_{\rho(r)}$ and $\marking^\ell_\rho(p_i)$ is the monomial $\ml{\dot{c}}{1}$. 
Because $\tplus_i(p_r)$ is the monomial $\ml{\dot{f}(X)}{1}$, we conclude $\marking^\prime(p_r)$ is the monomial $\ml{\dot{f}(\theta_{\rho(r)})}{1} = \ml{\theta_{\rho(r)+1}}{1}$. Because $\tdelta_i(q_i)$ is the monomial $\ml{\dot{c}}{-1}$, $\marking^\prime(q_i) = 0$. Because $\tdelta_i(q_z)$ is the monomial $\ml{\dot{c}}{1}$, $\marking^\prime(q_z)$ is the monomial $\ml{\dot{c}}{1}$. For all other $1\leq j\leq n$, $i\not= j\not=z$, $\marking(q_j) = \tdelta_i(q_j) = \marking^\prime(q_j) = 0$. Therefore, $\marking^\prime = \marking^{\rho^\prime}_{\ell^\prime}$.    
 \item Let $I_i = JZ(r,z_1,z_2)$. We distinguish between $\rho(r) = 0$ and $\rho(r) > 0$. 
 \begin{itemize}
  \item Let $\rho(r) = 0$. Then, $\marking^\ell_\rho(p_r)$ is the monomial $\ml{\theta_0}{1} = \ml{\dot{c}}{1}$. 
  Because $\tminusp_i(p_r)$ and $\tminusp_i(q_i)$ each are the monomial $\ml{\dot{c}}{1}$, $t_i^\prime$ is enabled in any firing mode $\sigma$. 
  Because $\tdeltap_i(p_r) = 0$, we conclude $\marking^\prime(p_r)$ is the monomial $\ml{\theta_0}{1} = \ml{\theta_{\rho^\prime(r)}}{1}$. 
  Because $\tdeltap_i(q_i)$ is the monomial $\ml{\dot{c}}{-1}$, $\marking^\prime(q_i) = 0$. 
  Because $\tdeltap_i(q_{z_2})$ is the monomial $\ml{\dot{c}}{1}$, $\marking^\prime(q_z)$ is the monomial $\ml{\dot{c}}{1}$. 
  For all other $1\leq j\leq n$, $i\not= j\not=z_2$, $\marking(q_j) = \tdeltap_i(q_j) = \marking^\prime(q_j) = 0$. Therefore, $\marking^\prime = \marking^{\rho^\prime}_{\ell^\prime}$.    
  \item Let $\rho(r) > 0$. Then, the proof is symmetrical to the case $I_i = INC(r,z_1)$. 
 \end{itemize}
\end{itemize}
\end{proof}

\begin{proof}[\ptit{step-minsky}]
 Inductively applying \lref{lem:minsky-step}, we conclude that for every state $(\rho,\ell)$ of $M$: $(0,1)\rightarrow^* (\rho,\ell)$ iff $\marking^\ell_\rho$ is reachable in $(N_M,\marking^1_0)$. Therefore, there exists $\rho$ with $(0,1)\rightarrow^* (\rho,n)$ iff marking $\marking^n_\rho$ is reachable. Obviously, $q_n = 0$ is valid iff for all reachable markings $\marking^\ell_\rho$: $\ell\not= n$. Hence, $q_n = 0$ is valid iff $M$ does not halt. 
\end{proof}

\begin{proof}[\ptit{markings-solutions}]
1.$\Rightarrow$2.:
Let $\Omega = \{ \omega \in \Theta \,|\, \exists p\in P : \exists \varrho\in\support{\marking(p)}:
\omega = \kterm_p\termMult\varrho $ and $\gamma_p\not=\gzero \}$.
For every $\omega \in \Omega$ we identify a \zero $\nu_\omega$ as follows:
$\nu_\omega(p) = \sum_{\omega=\kterm_p\termMult\varrho}\marking(p)(\varrho)$.
Furthermore, let $\marking_\omega$ be a marking with:
\[\marking_\omega(p)(\theta) =
\begin{cases}
\marking(p)(\theta) & \text{ if } \kterm_p\termMult\theta = \omega \\
0 & \text{ otherwise }. \\
\end{cases}\]
Additionally, let $\marking_0$ be the marking with $\marking_0(p) =
\begin{cases}
\marking(p) & \text{ if }\gamma_p=\gzero \\
0 & \text{ otherwise } 
\end{cases}$ for every $p\in P$.

We observe that $\marking_\omega(p)(\theta) \not = 0$ implies
$\marking_{\omega'}(p)(\theta) = 0$ for every $p\in 
P,\theta\in\Theta,\omega,\omega'\in\Omega,\omega'\not=\omega$.
Hence,
$\marking = \sum_{\omega \in \Omega}\marking_\omega + \marking_0$.
$\marking_0$ implements any \zero, especially the trivial \zero.
Next, we show that $\marking_\omega$ implements a \zero $\nu_\omega$,
defined by $\nu_\omega(p) = \sum_{\theta\in\Theta}\marking(p)(\theta)$.
First, for every $\nu_\omega$ we see: $0=(\kvector\termMult\marking)(\omega) =
\sum_{p\in P}\sum_{\omega=\theta_p\termMult\varrho}a_p\marking(p)(\varrho)$.
By the definition of $\nu_\omega$, the unification problem is solvable, as all terms
may unify to $\omega$.
Thus, each $\nu_\omega$ is
a \zero and each $\marking_\omega$ implements $\nu_\omega$.

\noindent 2.$\Rightarrow$1.: 
$\kvector\termMult\marking =
\kvector\termMult \sum_{1\leq i\leq n} \marking_i =
\sum_{1\leq i\leq n} \kvector\termMult\marking_i$.
As $\marking_i$ implements a \zero $\nu_i$,
we have $(\kvector\termMult\marking_i)(\varrho(\nu_i)) = 
\sum_{p\in P}\nu(p)\kcoeff_p = 0$.
And moreover, we have $\kvector\termMult\marking_i = 0$ and thus
$\kvector\termMult\marking = 0$.
\end{proof}

\begin{proof}[\ptit{realization-enabling}]
For this proof, let $\theta_{q,t} \in \Theta$ with $\{ \theta_{q,t} \}=\support{{\tminus(q)}}$ for all $q\in\pre{t}$.

 \noindent
 1.$\Rightarrow$2.:
 By \lref{lem:markings-solutions} there exist markings $\marking_1,\dots,\marking_n$ and \zeros
 $\nu_1,\dots,\nu_n$
 with $\marking = \sum_{i=1}^n\marking_i$ and $\marking_i$ implements zero $\nu_i$ (for $i=1,\dots,n$).
 For every $q\in \pre{t}$,
 there exists a $\marking_i$ with $\marking_i(q)(\val{\theta_{q,t}})\geq 1$.
 Moreover, $\sigma$ solves the unification problem for respective $\marking_i$.
 Hence, $\sigma$ is a realization of some derivation $\delta$.
 
 \noindent 2.$\Rightarrow$1.:
 As $\delta$ is derivable from $\spset$, there exist \zeros $\nu_q$ for
 every $q\in\pre{t}$ with $\nu_q(q)\geq 1$ and there exists a marking $\marking_q$ that implements
 $\nu_q$ with $\marking_q(q)(\val{\theta_{q,t}})\geq 1$ for $\sigma$.
 Then, $\marking = \sum_{q_\in\pre(t)}(\tminus(\theta_{q,t}))\marking_q(q)$ enables $t$.
 And, as for every $\marking_q$,
 $\kvector\termMult\marking_q=0$, we have $\kvector\termMult\marking=0$.
\end{proof}

\begin{proof}[\ptit{delta-zero-all-zero}]
1.$\Rightarrow$2.:
As $\sigma$ realizes $\delta$, there exists a $\sigma'$ with 
$\sigma(X)=\val[\sigma']{\delta(X)}$.
Thus, $\val{\tdelta} = \val[\sigma']{\subst[\delta]{\tdelta}}$.
Moreover, we have $\kvector\termMult\val[\sigma']{\subst[\delta]{\tdelta}}=
\val[\sigma']{\kvector\termMult\subst[\delta]{\tdelta}} = \val[\sigma']{0} = 0$.

\noindent 2.$\Rightarrow$1.:
Let $\sigma'$ be an assignment with $\val[\sigma']{\theta}=\val[\sigma']{\theta'}$ implies $\theta=\theta'$ for 
all $\theta,\theta'\in\support{\subst[\delta]{\tdelta}}$.
Such an assignment exists as $\Theta$ is infinite.
Then $\val[\sigma']{\subst[\delta]{\tdelta}}(\val[\sigma']{\theta}) =
\subst[\delta]{\tdelta}(\theta)$ for all $\theta\in\Theta$.
 Let $\sigma(X) = \val[\sigma']{\delta(X)}$ for all $X\in\varset$
 and $\kvector\termMult\val{\tdelta}=0$.
 Let $\theta\in\Theta$.
 Then $\kvector\termMult(\subst[\delta]{\tdelta})(\theta)=
 \kvector\termMult\val[\sigma']{\subst[\delta]{\tdelta}}(\val{\theta})=0$.
\end{proof}

\begin{proof}[\ptit{homo-stable-charac}]
1.$\Rightarrow$2.:
For all $\sigma$ with $\sigma$ is a realization of $\delta$, we have
by \lref{lem:realization-enabling} that there exists a marking $\marking$ with
$\kvector\termMult\marking=0$ and $\marking\geq\val{\tminus}$.
By stability of $E$, we have that $\kvector\termMult\val{\tdelta}=0$.
By \lref{lem:delta-zero-all-zero} this implies $\kvector\termMult(\subst{\tdelta})=0$.

\noindent 2.$\Rightarrow$1.:
Let $\marking$ be a marking and $\sigma$ be an assignment with: $\kvector\termMult\marking=0$
and $\marking\geq\val{\tminus}$.
By \lref{lem:markings-solutions}, $\sigma$ is realization of a $\delta$, that
is derivable from $\spset$.
As $\kvector\termMult(\subst[\delta]{\tdelta})=0$, by \lref{lem:delta-zero-all-zero}, we have that
$\kvector\termMult\val{\tdelta}=0$ and thus $E$ is stable.
\end{proof}

\begin{proof}[\ptit{finite-derivable}]
 Let $\delta$ be derivable from $\{ \nu_q \in \spset \cup \{\nu\} \mid q\in P \}$.
 Let $Q = \{ q\in P \mid \nu_q = \nu \}$. Let $q\in Q$.
 As $\nu$ is derivable, we have $\nu(q)\geq 1$.
 Thus, there exists an $i\leq n$ with $\nu_i(q)\geq 1$.
 As $\nu_i \leq \nu$, the unification problem of $\nu_i$ is a subset of the unification problem
 of $\nu$.
 Thus, there exists a substitution $\pi$ with $\subst[\pi]{\varrho(\nu_i)}=\varrho(\nu)$ and thus
 there exists a $\delta'$ that is derivable from $\spset$ and
 $\sigma$ is realizable from $\delta'$.
\end{proof}

\begin{proof}[\ptit{finitely-many-derivable}]
For every pre-place of $t$, we choose one \zero. Thus, the number of 
derivable $\delta$ is limited by:
 $|\{ \delta \mid \delta \text{ is derivable from }S \}| \leq |S|^{|\pre{t}|} < \infty$.
  The set may be computed by enumerating all possible candidates to the bound. 
\end{proof}

\begin{proof}[\ptit{given-finite-spanning-set-decidable}]
 By \lref{lem:finite-derivable}, for every $\delta$ that is derivable from the set of
 \zeros, there exists a $\delta^\prime$ derivable from $\spset$. By \lref{lem:finitely-many-derivable}, the set of all these $\delta^\prime$ is finite. 
 The proof of \lref{lem:finitely-many-derivable} also shows that the set may be computed by enumerating all possible candidates to the bound. 
 By \lref{lem:homo-stable-charac}, $E$ is stable if and only if for every
 $\delta$ we have $\kvector\termMult\subst[\delta]{\tdelta} = 0$, which is computable.
\end{proof}

\begin{proof}[\ptit{zero-bound}]
We denote: $\overline{P} :=  \{ p\in P | \gamma_p > 0 \}$,
$\underline{P} :=  \{ p\in P | \gamma_p < 0 \}$,
$\underline{\nu} := \max\{ \nu(p) | p\in\underline{P}\}$ and
$\overline{\nu} := \max\{ \nu(p) | p\in\overline{P}\}$. 
As $\nu$ is a zero, we have
  $\sum_{p\in P}\nu(p) =  \sum_{p\in\underline{P}}\nu(p) + \sum_{p\in\overline{P}}\nu(p)$
  and hence $\sum_{p\in\underline{P}}\nu(p) = \sum_{p\in\overline{P}}\nu(p)$ (*).
We distinguish two cases whether 1. is true or 2. is true.
\begin{enumerate}
 \item 
  $\sum_{p\in P}\nu(p) =  \sum_{p\in\underline{P}}\nu(p) + \sum_{p\in\overline{P}}\nu(p)
  \leq  \sum_{p\in\underline{P}}|\gamma_p|\nu(p) + \sum_{p\in\overline{P}}\gamma_p\nu(p)
  \overset{(*)}{=}  2\sum_{p\in\underline{P}}|\gamma_p|\nu(p)
  = 2\sum_{p\in\underline{P}}\underline{\gamma}\underline{\nu}
  \overset{(1.)}{<}  2\sum_{p\in\underline{P}}\underline{\gamma}\overline{\gamma}
  =  2|\underline{P}|\underline{\gamma}\overline{\gamma}
  \leq  2|P|\underline{\gamma}\overline{\gamma}$
  \item 
  $\sum_{p\in P}\nu(p) =  \sum_{p\in\underline{P}}\nu(p) + \sum_{p\in\overline{P}}\nu(p)
  \leq  \sum_{p\in\underline{P}}|\gamma_p|\nu(p) + \sum_{p\in\overline{P}}\gamma_p\nu(p)
  \overset{(*)}{=}  2\sum_{p\in\overline{P}}\gamma_p\nu(p)
  = 2\sum_{p\in\overline{P}}\overline{\gamma}\overline{\nu}
  \overset{(2.)}{<}  2\sum_{p\in\overline{P}}\overline{\gamma}\underline{\gamma}
  =  2|\overline{P}|\underline{\gamma}\overline{\gamma}
  \leq  2|P|\underline{\gamma}\overline{\gamma}$
  
\end{enumerate}

\end{proof}


\end{document}